\documentclass[preprint,12pt]{elsarticle}
\usepackage{amsmath}
\usepackage{amssymb}
\usepackage{amsthm}
\usepackage{url}
\usepackage[dvipsnames]{color}
\usepackage{hyperref}

\newcommand{\bbK}{\ensuremath{\mathbb{K}}}
\newcommand{\bbC}{\ensuremath{\mathbb{C}}}
\newcommand{\modmu}{\;(\mathrm{mod}\,\mu)}
\newcommand{\ee}{\mathrm{e}} %exponent
\newcommand{\dd}{\mathrm{d}} %differential
\newcommand{\ii}{\mathrm{i}}
\newcommand{\hhv}{\mathbf h}
\newcommand{\ez}{\mathbf e_3}

\newtheorem{theorem}{Theorem}[section]
\newtheorem{corollary}{Corollary}[theorem]
\newtheorem{proposition}{Proposition}[section]
\newtheorem{lemma}[theorem]{Lemma}
\newtheorem{definition}{Definition}[section]
\newtheorem{remark}{Remark}

\begin{document}

\begin{frontmatter}

\title{Recursive construction of spectral basis and its application to matrices}

\author{A.~Acus}
\address{Institute of Theoretical Physics and Astronomy,\break
Vilnius University,\break Saul{\.e}tekio 3, LT-10257 Vilnius,
Lithuania, arturas.acus@tfai.vu.lt}

\author{A.~Dargys}

\address{%
Center for Physical Sciences and Technology,\break Semiconductor
Physics Institute,\break Saul{\.e}tekio 3, LT-10257 Vilnius,
Lithuania, adolfas.dargys@ftmc.lt}

\begin{abstract}
We prove a recursive formula for construction of generalized spectral basis.
 Example of application is presented for symbolic computation with non-diagonalizable
 matrix that describes electromagnetic wave propagation in anisotropic media. The method is general enough and can be applied,
 for example, to multivectors of the Clifford algebra.
\end{abstract}


\end{frontmatter}

\section{Introduction}

If an element of  associative algebra is diagonalizable then an analytical function be expressed in terms of ordinary spectral projectors. However, in the non-diagonalizable case, repeated roots generate nilpotent corrections, and the ordinary spectral decomposition must  be
replaced by more general spectral basis~\cite{Sobczyk1996,Sobczyk1997}. The practical motivation is simple. If minimal polynomial computed for algebra element $A$, for
example, Clifford number (multivector) or matrix, factors out as
\begin{equation}
\mu(x)=\prod_{i=1}^s (x-\lambda_i)^{m_i};\quad n=\deg \mu(x)=\sum_{i=1}^s m_i
\label{MinimalPolynomial}
\end{equation}
where the number of factors $s$ matches the number of distinct roots $\lambda_i$ having multiplicity $m_i$. Then each multiple root contributes
not only an idempotent projector but also a finite chain of nilpotent operators of length $m_i-1$. In symbolic computations, one often needs an explicit way to construct~\cite{AcusDargys2025} the generalized spectral basis for algebra element $A$ directly from $\mu$, without the use of Jordan matrix form. An example could be Clifford algebra multivector $A$, for which we want to avoid computation of the matrix representation. In the following we assume that polynomial $\mu(x)$, which is associated with  element $A$, is a
minimal polynomial.  For computation of function of $A$, this, however, is not strictly required and an arbitrary polynomial is appropriate. Of course, a polynomial of higher degree would require computational overheads, since the higher powers of $A$ then will appear.

In Sec.~\ref{sec2}, the terminology and the procedure are  introduced.  Sec.~\ref{sec3} and \ref{sec4}, respectively, explain how to find a
generalized spectral basis by recursion and clarifies its structure. A simple  example and concluding remarks are presented in Sec.~\ref{sec6}
and~\ref{sec7}.

\section{General settings and notations}\label{sec2}
We will work in the quotient algebra $ \bbK[x]/\mu$, \textit{i.e.} in the algebra of polynomials modulo the ideal generated by $\mu(x)$. Remind that, two polynomials $f,g\in \bbK[x]$ represent the same element of $\bbK[x]/\mu$ precisely when $f(x)-g(x)$ is divisible by $\mu(x)$. Throughout, we write $f(x)\equiv g(x)\modmu$ to mean equality in this quotient. We assume that $\bbK$ is a field of complex numbers $\bbC$.

The quotient is the natural polynomial model for functional calculus based on the minimal polynomial. If an algebra element $A$ satisfies $\mu(A)=0$, then any polynomial identity modulo $\mu(x)$ immediately becomes a true identity for the algebra element $A$ after substituting $x\mapsto A$. This property enables straightforward computation of function of $A$.

\begin{definition}[Primary and complementary factors]
For each root $\lambda_i$ of multiplicity $m_i$ the factors in 
\begin{equation}\label{PrimaryFactors}
\mu_i(x):=(x-\lambda_i)^{m_i}\quad \textrm{and}\quad 
M_i(x):=\frac{\mu(x)}{\mu_i(x)}=\prod_{\substack{j=1\\j\neq i}}^s (x-\lambda_j)^{m_j}
\end{equation}
are respectively primary and complementary factors of the polynomial $\mu(x)$.
\end{definition}
Since the roots are distinct, $\mu_i$ and $M_i$ are coprime, \textit{i.e.} have no nontrivial common factors. 

The existence of two kinds of factors enables straightforward construction of spectral projectors onto the local natural basis
$1, (x-\lambda_i), (x-\lambda_i)^2, \dots, (x-\lambda_i)^{m-1}$ of root $\lambda_i$. The projection is realised by operators $Q_{i,r}(x)\equiv (x-\lambda_i)^r P_i(x)\modmu$, where $P_i(x):=Q_{i,0}(x)$.

\begin{proposition}[Primary projector]
There exists a unique polynomials $P_i(x)\in \bbK[x]/\mu$ such that
  \begin{equation}\label{PrimaryProjector}
    P_i(x)\equiv 1 \pmod{(x-\lambda_i)^{m_i}},
\qquad
P_i(x)\equiv 0 \pmod{M_i(x)}.
  \end{equation}
\end{proposition}

\begin{proof}
The proof is based on Chinese remainder theorem.  Since $(x-\lambda_i)^{m_i}$ and $M_i(x)$ are coprime, the Chinese remainder theorem claims that $\mu$ divides a polynomial when it is divisible by both, primary and complementary, factors: $\bbK[x]/\mu \cong \bbK[x]/\bigl((x-\lambda_i)^m\bigr)\times \bbK[x]/(M_i)$.
Therefore, there exists a unique class modulo $\mu$ mapping to $(1,0)$ for each root $\lambda_i$. This is exactly the required polynomial $P_i(x)$.
\end{proof}

\begin{corollary}\label{projectionOperatorProperties}
  The projectors $P_i(x)$ have properties
\begin{align}
  &P_i^2(x)\equiv P_i(x)\modmu&\\
  &P_i(x)P_j(x)\equiv 0\modmu&\textrm{if}\quad \lambda_i\neq \lambda_j\\
  &\sum_{i=1}^s P_i(x)\equiv 1\modmu.
  \end{align}
\end{corollary}
\begin{proof}
  Modulo $(x-\lambda_i)^{m_i}$, one has $P_i\equiv 1$, so $P_i^2\equiv 1\equiv P_i$. Modulo $M_i$, one has $P_i\equiv 0$, so $P_i^2\equiv 0\equiv P_i$. By uniqueness modulo $\mu$, this implies $P_i^2\equiv P_i\modmu$.
The orthogonality and partition of unity follow similarly.
\end{proof}

\begin{definition}[Generalized spectral basis]\label{defGenSpectralBasis}
  The set of $n=\sum_{i=1}^s m_i$ polynomials 
\begin{equation}
  \left\{
(\lambda_i,\{P_i(x):= Q_{i,0},Q_{i,1},\dots,Q_{i,m_i-1}\})
\right\}_{i=1}^s
\end{equation}
constitute a generalized spectral basis of $\mu(x)$ if they have the properties:
\begin{enumerate}
  \item For different $i\neq j$ indices polynomials $Q_{i,r}$ are pairwise orthogonal  
    \begin{align}
  Q_{i,r}(x)Q_{j,\ell}(x)\equiv 0\modmu&&\text{for all admissible}\quad r,\ell 
\end{align}
$P_i(x)P_j(x)\equiv 0\modmu$ representing a particular case.

\item $P_i= Q_{i,0}$ realises partition of unity:
  \begin{align}
  \sum_{i=1}^s P_i(x)\equiv 1\modmu
\end{align}

\item $Q_{i,r}$ satisfy the multiplication rules:
  \begin{align}
    &Q_{i,r}(x)Q_{i,\ell}(x)\equiv Q_{i,r+\ell}(x)\modmu &&\textrm{if}\quad r+\ell\le m-1\notag \\
    &Q_{i,r}(x)Q_{i,\ell}(x)\equiv 0\modmu &&\textrm{if}\quad r+\ell\ge m 
\end{align}
    including $Q_{i,0}^2(x)\equiv Q_{i,0}(x)\modmu$, {\it i.e.}  $P_{i}^2(x)\equiv P_{i}(x)\modmu$. 
\end{enumerate}
\end{definition}
\begin{remark}
The polynomials $Q_{i,r}(x)$ when $r\ge 1$ are called~\cite{Sobczyk1997} nilpotents (degree of nilpotency $m_i$) that are projectively related to mutually annihilating idempotents $P_k(x)=Q_{k,0}(x)$.
\end{remark}

\section{Recursive computation of generalised spectral basis}
\label{sec3}
Now we present a constructive way to compute the generalised spectral basis.
\begin{theorem}[Recursive generation of generalized spectral basis]
\label{MainTheorem}
Let
\[
\mu(x)=\prod_{i=1}^s (x-\lambda_i)^{m_i}=\sum_{k=0}^{n} c_{n-k}\, x^k
\]
  be a monic, $c_{0}=1$, polynomial over field of characteristic zero. 
  Then a generalized spectral basis can be computed by the following recursive algorithm:
\begin{enumerate}
\item Define the bivariate polynomial $S(x,\lambda)$ 
\[
S(x,\lambda)
=
\sum_{r=0}^{n-1}
\left(
\sum_{k=0}^{n-1-r} c_{n-1-r-k}\lambda^k
\right)x^r.
\]
\item
For each root $\lambda_i$ of multiplicity $m_i$ compute recursively
\[
Q_{i,m_i-1}(x)=\frac{S^{[0]}(x,\lambda_i)}{\mu^{[m_i]}(\lambda_i)}
\]
and for each $r=1,\dots,m_i-1$
\[
Q_{i,m_i-1-r}(x)
=
\frac{
S^{[r]}(x,\lambda_i)
-
\sum_{k=1}^{r}Q_{i,m_i-1-r+k}(x)\mu^{[m_i+k]}(\lambda_i)
}{
\mu^{[m_i]}(\lambda_i)
}
\]
where $\mu^{[r]}(\lambda)=\frac{1}{r!}\frac{d^r}{d\lambda^r}\mu(\lambda)$ and  
    $S^{[r]}(x,\lambda)=\frac{1}{r!}\frac{\partial^r}{\partial \lambda^r}S(x,\lambda)$ for $r\ge 0$ denote normalized derivatives and where $\mu^{[0]}(\lambda)=\mu(x)$, $S^{[0]}(x,\lambda)=S(x,\lambda)$.

\item The resulting family 
\[
\left\{
(\lambda_i,\{Q_{i,0},Q_{i,1},\dots,Q_{i,m_i-1}\})
\right\}_{i=1}^s
\]
constitute the generalized spectral basis.
\end{enumerate}
\end{theorem}
\begin{remark} Note that for root $\lambda_i$ of multiplicity $m_i$ we have 
$\mu^{[0]}(\lambda_i)=\mu^{[1]}(\lambda_i)=\cdots=\mu^{[m_i-1]}(\lambda_i)=0$, but  
$\mu^{[m_i]}(\lambda_i)\neq 0$.
This is the usual characterization of multiplicity: a root has multiplicity $m_i$ precisely when the first $m_i-1$ derivatives vanish and the $m_i$-th does not. The property also ensures that in the computation process denominators of the expressions in the theorem never vanish. 
\end{remark}
\begin{remark}
The recursion proceeds from the \emph{highest nilpotent level downward}. Reverse order is expected at first sight, but it is natural: the coefficient \(\mu^{[m]}(\lambda)\) is the first nonzero derivative at the repeated root, so it determines the top term, and once that is known, lower terms are recovered one by one.
\end{remark}
We will prove the theorem in few steps. First we will establish useful relations between polynomials $\mu(x)$ and $S(x,\lambda)$ and their derivatives that will help to write down their Taylor expansion. Then using the recurrent relations and the method of induction we will identify Taylor expansion terms with idempotents $P_{i}$ and nilpotents $Q_{i,r}$ that have properties listed in {\bf Definition}~\ref{defGenSpectralBasis}. 

Now we introduce a new expression for polynomial $S(x,\lambda)$. To shorten notation we temporarily remove the root index, {\it i.e.} we simply write $\lambda$ instead of $\lambda_i$, $m$ in the place of $m_i$, and replace $Q_{i,r}$ by $Q_{r}$.
\begin{lemma}\label{SmuLemma}
The polynomial $S(x,\lambda)$ can be written as  
  \begin{equation}\label{SmuRelation}
    S(x,\lambda):=\frac{\mu(x)-\mu(\lambda)}{x-\lambda}.
  \end{equation}
\end{lemma}
\begin{proof}
Taking into account explicit form of polynomial $\mu(x)=c_0 x^n+c_1 x^{n-1}+\cdots+c_{n-1} x+c_n$, 
and  the expression 
$S(x,\lambda) =\sum_{r=0}^{n-1}\left(
\sum_{k=0}^{n-1-r} c_{n-1-r-k}\lambda^k
\right)x^r$, we establish the relation 
  \begin{equation}\label{SmuRelation1}
\mu(x)-\mu(\lambda)=(x-\lambda)S(x,\lambda).
 \end{equation}
Indeed, this explicit form follows from the algebraic identity
\[
x^m-\lambda^m=(x-\lambda)\sum_{k=0}^{m-1} x^{m-1-k}\lambda^k .
\]
Applying it term-by-term and collecting the powers of $x$ gives the relation between polynomials $\mu(x)$ and $S(x,\lambda)$. Solving $S(x,\lambda)$ yields the answer.
\end{proof}

Next, we need an expression relating $Q_{r}$ and normalized derivatives of polynomials $\mu(x)$ and $S(x,\lambda)$. This relation between them plays a key role in the recursion procedure and the induction step.

\begin{lemma}[Triangular identity]
  For $r=0,1,\ldots,m-1$ recursion formulas in step~2 of the {\bf Theorem}~\ref{MainTheorem} can be compactly rewritten as 
  \begin{equation}\label{forInduction}
S^{[r]}(x,\lambda)=\sum_{k=0}^{r} Q_{m-1-r+k}(x)\,\mu^{[m+k]}(\lambda).
\end{equation}
\end{lemma}
\begin{proof}
For $r=0$ the formula for $Q_{m-1}$ is $S^{[0]}(x,\lambda)=Q_{m-1}(x)\mu^{[m]}(\lambda)$. 
For $r\ge 1$, multiply the recursive formula by $\mu^{[m]}(\lambda)$ and move the sum to the other side. This gives
\[
S^{[r]}(x,\lambda)
=
Q_{m-1-r}(x)\mu^{[m]}(\lambda)
+
\sum_{k=1}^{r}Q_{m-1-r+k}(x)\mu^{[m+k]}(\lambda),
\]
  which is precisely the triangular identity \eqref{forInduction}, that expresses each normalized derivative $S^{[r]}$ as a linear combination of the unknown basis polynomials $Q_j$, with coefficients of higher normalized derivatives of $\mu(\lambda)$.
\end{proof}

Now we are ready to write down Taylor expansion of $\mu(x)$ and $S(x,\lambda)$ around root of multiplicity $m$.
\begin{lemma}[Taylor series of $\mu$ at a repeated root]
For a root $\lambda$ of multiplicity $m$,
  \begin{equation}\label{TaylorMu}
\mu(x)=\sum_{j=m}^{n}\mu^{[j]}(\lambda)(x-\lambda)^j
=
(x-\lambda)^m \sum_{r=0}^{n-m}\mu^{[m+r]}(\lambda)(x-\lambda)^r.
\end{equation}
\begin{proof}
This is simply the Taylor expansion of $\mu$ at $\lambda$, together with the fact that the first $m-1$ normalized derivatives vanish.
\end{proof}
\end{lemma}

\begin{corollary}
From Eq.~\eqref{PrimaryFactors} for complementary factor  we have
\[
M_\lambda(x):=\frac{\mu(x)}{(x-\lambda)^m} = \sum_{r=0}^{n-m}\mu^{[m+r]}(\lambda)(x-\lambda)^r .
\]
\end{corollary}

\begin{lemma}[Expansion of derivatives of $S(x,\lambda)$ around a repeated root]
For every $r\ge 0$,
  \begin{equation}\label{derivativeSa}
S^{[r]}(x,\lambda) = \sum_{j=r+1}^{n}\mu^{[j]}(\lambda)(x-\lambda)^{j-r-1}
  \end{equation}
or equivalently,
  \begin{equation}\label{derivativeSb}
S^{[r]}(x,\lambda)=\sum_{t\ge 0}\mu^{[r+1+t]}(\lambda)(x-\lambda)^t.
  \end{equation}
\end{lemma}

\begin{proof}
Using Eq.~\eqref{SmuRelation1} we expand $\mu(x)$ in Taylor series at $\lambda$ using~\eqref{TaylorMu}. Then 
  subtraction of $\mu(\lambda):=\mu^{[0]}(\lambda)$ and division by $(x-\lambda)$ gives
\[
S(x,\lambda)=\sum_{j\ge 1}\mu^{[j]}(\lambda)(x-\lambda)^{j-1}.
\]
Taking normalized derivatives with respect to $\lambda$ yields~\eqref{derivativeSa}. 
Reindexing gives~\eqref{derivativeSb}.
\end{proof}
\begin{remark}
  Formula \eqref{derivativeSb} is very important. It says that the $r$-th normalized derivative $S^{[r]}$ starts with the coefficient $\mu^{[r+1]}(\lambda)$, and continues with higher powers of $(x-\lambda)$. Since $\mu^{[1]}(\lambda),\ldots,\mu^{[m-1]}(\lambda)$ vanish, this means that each $S^{[r]}$ begins only at the relevant primary level.
\end{remark}

Now we restore root index $\lambda_i$, identify Taylor expansion terms with polynomials $Q_{i,r}$, and show that the recursive construction recovers the expected basis in the $\lambda_i$-primary component.

\begin{lemma}[Identification of $Q_{i,r}$ with powers of $x-\lambda_i$]\label{identificationLemma}
Let $P_i(x)$ be the projector onto the $\lambda_i$-primary component. Then the recursively defined polynomials satisfy
  \begin{equation}\label{identificationFormula}
    Q_{i,r}(x)\equiv (x-\lambda_i)^r P_i(x)\modmu,
\qquad r=0,1,\dots,m_i-1.
\end{equation}
\end{lemma}
\begin{proof}
We prove the statement by induction on $r$, starting from $r=m_i-1$.
\medskip
\noindent
\textbf{Verification step for $r=m_i-1$.}\newline
By definition $Q_{i,m_i-1}(x)=\frac{S^{[0]}(x,\lambda_i)}{\mu^{[m_i]}(\lambda_i)}$. 
Using formula \eqref{derivativeSb} with $r=0$,
\[
S^{[0]}(x,\lambda_i)
=
\mu^{[m_i]}(\lambda_i)(x-\lambda_i)^{m_i-1}
+
\mu^{[m_i+1]}(\lambda_i)(x-\lambda_i)^{m_i}
+\cdots.
\]
Dividing by \(\mu^{[m_i]}(\lambda_i)\), we obtain
\[
  Q_{i,m_i-1}(x)\equiv (x-\lambda_i)^{m_i-1}\pmod{(x-\lambda_i)^{m_i}}.
\]
Also, by construction $Q_{i,m_i-1}$ lies in the $\lambda_i$-primary component, hence it vanishes modulo the complementary factor $M_i(x)$. Therefore
\[
Q_{i,m_i-1}(x)\equiv (x-\lambda_i)^{m_i-1}P_i(x)\modmu.
\]

\medskip
\noindent
\textbf{Induction step.}
Assume that for some $r\in\{1,\ldots,m_i-1\}$,
\[
  Q_{i,m_i-r}(x)\equiv (x-\lambda_i)^{m_i-r}P_i(x),\quad
Q_{i,m_i-r+1}(x)\equiv (x-\lambda_i)^{m_i-r+1}P_i(x),\quad \ldots
\]
up to $Q_{i,m_i-1}(x)$. We must prove
\[
Q_{i,m_i-1-r}(x)\equiv (x-\lambda_i)^{m_i-1-r}P_i(x)\modmu.
\]

The recursion in {\bf Theorem}~\ref{MainTheorem} is defined as
$$
Q_{i,m_i-1-r}(x) =
\frac{
  S^{[r]}(x,\lambda_i)-\sum_{k=1}^{r}Q_{i,m_i-1-r+k}(x)\mu^{[m_i+k]}(\lambda_i)
}{\mu^{[m_i]}(\lambda_i)}.
$$
Now use expansion~\eqref{derivativeSb},
\[
S^{[r]}(x,\lambda_i)=\mu^{[m_i]}(\lambda_i)(x-\lambda_i)^{m_i-1-r}
+
\sum_{k=1}^{\infty}\mu^{[m_i+k]}(\lambda_i)(x-\lambda_i)^{m_i-1-r+k}.
\]
By induction hypothesis,
\[
Q_{i,m_i-1-r+k}(x)\equiv (x-\lambda_i)^{m_i-1-r+k}P_i(x)\modmu.
\]
We see that the subtraction in the numerator removes precisely the higher-order terms coming from
\[
\mu^{[m_i+k]}(\lambda_i)(x-\lambda_i)^{m_i-1-r+k},
\qquad k=1,\ldots,r.
\]
What remains is
$$
\mu^{[m_i]}(\lambda_i)(x-\lambda_i)^{m_i-1-r}P_i(x)\quad \textrm{modulo}\quad \mu.
$$ 
Dividing by $\mu^{[m_i]}(\lambda_i)\neq 0$ gives
\[
Q_{i,m_i-1-r}(x)\equiv (x-\lambda_i)^{m_i-1-r}P_i(x)\modmu.
\]
This completes the induction.
\end{proof}
\begin{remark}
This lemma is the centerpiece of the theorem proof. The recursion looks somewhat complicated, but it simply reconstructs the natural powers
\[
P_i,\quad (x-\lambda_i)P_i,\quad (x-\lambda_i)^2P_i,\quad \ldots,\quad (x-\lambda_i)^{m_i-1}P_i
\]
in a way that is compatible with reduction modulo $\mu$.
\end{remark}

\section{Multiplicative structure of generalised spectral basis}\label{sec4}

Once identification with projector operators~\eqref{identificationFormula} is established, the multiplication rules for $Q_{i,r}(x)$ become transparent. 

\begin{proposition}[Multiplication rules for the same root]
  Let $Q_{i,0},\ldots,Q_{i,m_i-1}$ be the recursively constructed polynomials associated with the root $\lambda_i$. Then:
  \begin{enumerate}[I.]
\item $Q_{i,0}^2(x)\equiv Q_{i,0}(x)\modmu$
\item $Q_{i,0}(x)Q_{i,r}(x)\equiv Q_{0,r}(x)\modmu,\qquad r=0,\ldots,m_i-1$
\item
  If $r+\ell\le m_i-1$, then $Q_{i,r}(x)Q_{i,\ell}(x)\equiv Q_{i,r+\ell}(x)\modmu$
\item If $r+\ell\ge m_i$, then $Q_{i,r}(x)Q_{i,\ell}(x)\equiv 0\modmu$
\end{enumerate}
\end{proposition}

\begin{proof}
We prove  most general multiplication rules (iii) and (iv), since the statements (i) and (ii) are particular cases of (iii). After substitution of~\eqref{identificationFormula} we get
\[
  Q_{i,r}(x)Q_{i,\ell}(x)
\equiv
(x-\lambda_i)^rP_i(x)\,(x-\lambda_i)^\ell P_i(x)
=
(x-\lambda_i)^{r+\ell}P_i^2(x).
\]
Since $P^2_i\equiv P_i\modmu$, we have $Q_{i,r}(x)Q_{i,\ell}(x)\equiv (x-\lambda_i)^{r+\ell}P_i(x)\modmu$.
\newline\noindent
If $r+\ell\le m_i-1$, then {\bf Lemma}~\ref{identificationLemma} gives
\[
(x-\lambda_i)^{r+\ell}P_i(x)\equiv Q_{i,r+\ell}(x)\modmu.
\]
  If \(r+\ell\ge m_i\), then $(x-\lambda_i)^{r+\ell}$ is divisible by $(x-\lambda_i)^{m_i}$, hence
\[
  (x-\lambda_i)^{r+\ell}P_i(x)\equiv 0\modmu.
\]
\end{proof}

\begin{proposition}[Orthogonality between different roots]
If $i\neq j$, then
\[
Q_{i,r}(x)Q_{j,\ell}(x)\equiv 0\modmu
\]
for all admissible \(r,\ell\).
\end{proposition}

\begin{proof}
From expression~\eqref{identificationFormula}, 
$
Q_{i,r}(x)\equiv (x-\lambda_i)^r P_i(x)\modmu$ and 
$
Q_{j,\ell}(x)\equiv (x-\lambda_j)^\ell P_j(x)\modmu$.
Therefore,
\[
Q_{i,r}(x)Q_{j,\ell}(x)
\equiv
(x-\lambda_i)^r(x-\lambda_j)^\ell P_i(x)P_j(x)\modmu.
\]
By {\bf Corollary}~\ref{projectionOperatorProperties}, $P_iP_j\equiv 0\modmu$ for $i\neq j$, so the product vanishes.
Since $P_i$ are projection operators, the corollary also guarantee partition of unity property.
\end{proof}

This finishes the proof of the {\bf Theorem}~\ref{MainTheorem}, since we explicitly checked that the constructed polynomials $Q_{i,r}$ satisfy all properties of a generalised spectral basis. 

\subsection{Interpretation for functional calculus and application to computation of analytic function}
\label{sec4b}

Apart from being algebraically elegant the {\bf Theorem}~\ref{MainTheorem} is exactly what is needed in practice.
Suppose $A$ is an element of an associative algebra such that $\mu(A)=0$.
Then we may substitute \(x\mapsto A\) into all polynomials above, and define
\[
  P_i:=Q_{i,0}(A),\qquad N_{i,r}:=Q_{i,r}(A)\quad \textrm{for}\quad r\ge 1.
\]
The theorem implies
\[
\sum_{i=1}^s P_i = 1,
\qquad
P_iP_j=\delta_{ij}P_i
\]
and within each $i$-th block,
\begin{align}
&P_i N_{i,r}=N_{i,r},\quad
  N_{i,r}N_{i,\ell}=N_{i,r+\ell}&&\textrm{for}\quad r+\ell\le m_i-1,\\
  &N_{i,r}N_{i,\ell}=0&&\textrm{for}\quad r+\ell\ge m_i.
\end{align}

Thus the recursively generated polynomials recover, in a coordinate-free way, the same structure that one ordinarily associates with Jordan blocks.

If $f(x)$ is analytic in a neighborhood of the spectrum, then one obtains~\cite{Sobczyk1997} the standard generalized spectral formula
\begin{equation}\label{FormulaForFunction}
f(A)
=
\sum_{i=1}^s
\sum_{r=0}^{m_i-1} f^{[r]}(\lambda_i)\,Q_{i,m_i-1-r}(A).
\end{equation}
For \(f(x)=e^x\), this becomes
\[
e^A
=
\sum_{i=1}^s e^{\lambda_i}
\sum_{r=0}^{m_i-1}\frac{1}{r!}Q_{i,m_i-1-r}(A).
\]
This is precisely why the recursive construction is so useful computationally. It transforms a non-diagonalizable functional calculus problem into a finite algebraic computation with polynomials and derivatives.

\section{Application to defective $6\times6$ Maxwell Equation}\label{sec6}
In physics and engineering, non-diagonal matrices have found a new and promising application in micro- and nanophotonic devices. The spectral singularities here are called exceptional points \cite{Miri2019}. In the model described below a defective matrix arises naturally, without imposing special conditions on model parameters.

\subsection{Model formulation}
Consider a source-free homogeneous medium characterized by field relations
\begin{equation*}
  \mathbf D=\varepsilon_0\boldsymbol{\varepsilon}_r\mathbf E,
  \quad
  \mathbf B=\mu_0\boldsymbol{\mu}_r\mathbf H
\end{equation*}
where $\mathbf D$ and $\mathbf B$ are electric and magnetic inductions, whereas $\mathbf E$ and $\mathbf H$ are electric and magnetic fields. Assume that electrical relative-permittivity inverse tensor is
\begin{equation}\label{Permittivity}
  \boldsymbol{\varepsilon}_r^{-1}=\mathbf A=
  \begin{pmatrix}
    \alpha & g & 0\\
    0 & \alpha & 0\\
    0 & 0 & \beta
  \end{pmatrix},
  \qquad \alpha\beta g\ne0
\end{equation}
and relative magnetic permeability tensor  $\boldsymbol{\mu}_r$ is a unit matrix. The off-diagonal entry describes rotation of electric field
polarization in $x-y$ plane. Eq.~\eqref{Permittivity} should be regarded as an idealized  model of anisotropic dielectric.

Let's introduce normalized fields,
\begin{equation*}
  \mathbf{e}=\sqrt{\varepsilon_0}\,\mathbf E,
  \qquad
  \hhv=\sqrt{\mu_0}\,\mathbf H
\end{equation*}
and consider a plane wave  $\exp(\ii kz)$ that propagates in $z$-direction and is characterised by polarization vector $\mathbf{v}$.  Define polarization
matrix
\begin{equation*}
  \mathsf C=
  \begin{pmatrix}
    0&-1&0\\
    1& 0&0\\
    0& 0&0
  \end{pmatrix},
  \quad
  \mathsf C\mathbf v=\ez\times\mathbf v,\quad\mathbf{e}_3||z\, .
\end{equation*}
Then  electromagnetic field propagation equation takes a form of a system of coupled Schr\"odinger-like equations,
\begin{equation}\label{ModelEquations}
  \ii\frac{\partial}{\partial t}
  \Bigl(\begin{matrix}\mathbf{e}\\\,\hhv\,\end{matrix}\Bigr)
  =H_M
  \Bigl(\begin{matrix}\mathbf{e}\\\,\hhv\,\end{matrix}\Bigr),
  \qquad
  H_M=ck
  \begin{pmatrix}
    0&-\mathbf A\mathsf C\\
    \mathsf C&0
  \end{pmatrix}.
\end{equation}
Explicitly, in the ordered basis $(e_x,e_y,e_z,h_x,h_y,h_z)^{\mathsf T}$, the matrix and vector are
\begin{equation}\label{ModelHamiltonian}
  H_M=ck
  \begin{pmatrix}
    0&0&0&-g&\alpha&0\\
    0&0&0&-\alpha&0&0\\
    0&0&0&0&0&0\\
    0&-1&0&0&0&0\\
    1&0&0&0&0&0\\
    0&0&0&0&0&0
  \end{pmatrix},\quad H_M\Bigl(\begin{matrix}\mathbf{e}\\\,\hhv\,\end{matrix}\Bigr)=
 ck \begin{pmatrix}(\alpha h_y-g h_x)\sqrt{\mu_0}\\-\alpha h_x\sqrt{\mu_0}\\0\\-e_y\sqrt{\varepsilon_0}\\-e_x\sqrt{\varepsilon_0}\\0 \end{pmatrix}.
\end{equation}

The minimal polynomial of $H_M$, {\it i.e.} a polynomial $\mu(x)$ of smallest degree, which becomes zero matrix upon substitution $x\mapsto H_M$, is 
  \begin{equation}\label{ModelMinPoly}
  \mu_{\scriptscriptstyle H_M}(\lambda)
  =\lambda\bigl(\lambda^2-(ck)^2\alpha\bigr)^2\, .
\end{equation}
Thus, the matrix $H_M$ is non-diagonalizable. Eigenvalues of~\eqref{ModelMinPoly} are
\begin{equation}
  \lambda_0=0,
  \qquad
  \lambda_\pm=\pm ck\sqrt{\alpha}\,.
\end{equation}
Each of $\lambda_+$ and $\lambda_-$ eigenvalues have algebraic multiplicity two that in physics represents two possible EM field rotation
directions in $x-y$ plane. If $g\ne0$, each has only a single linearly independent eigenvector.  So the matrix $H_M$ can be reduced to a direct
sum,
\begin{equation*}
  H_M\sim
  J_2\!\left(ck\sqrt{\alpha}\right)
  \oplus
  J_2\!\left(-ck\sqrt{\alpha}\right)
  \oplus[0]\oplus[0]
\end{equation*}
where $J_2(..)$ is $2\times 2$ Jordan block that signify two possible field rotation directions. 
The two longitudinal zero modes are the usual Maxwell constraint modes, while the defective Jordan blocks lie in the physical
transverse-polarization sector.

\subsection{Solution by computing exponential function of matrix}
Solution of the model equations~\eqref{ModelEquations} can be written in a symbolic form as an exponential of matrix
\begin{equation*}
  \begin{pmatrix}\mathbf{e}(t)\\\,\hhv(t)\,\end{pmatrix}
  =\exp(-\ii H_Mt)
  \begin{pmatrix}\mathbf{e}(0)\\\,\hhv(0)\,\end{pmatrix}.
\end{equation*}
We will demonstrate how to compute this exponential using generalised spectral basis technique given by 
{\bf Theorem}~\ref{MainTheorem}. 

First compute minimal polynomial~\cite{Bialas2008} of the $-\ii t H_M$, where the matrix $H_M$ is defined by~\eqref{ModelHamiltonian}. We have
$\mu(x)=x \left(\alpha  c^2 k^2 t^2+x^2\right)^2=\alpha ^2 c^4 k^4 t^4 x+2 \alpha  c^2 k^2 t^2 x^3+x^5$. It has three distinct roots $\lambda_1,\lambda_2,\lambda_3$ of multiplicities~$m_1=1$, $m_2=2$, and $m_3=2$ respectively,
\[
  \begin{aligned}  
    &\text{root,}\ \lambda_i&\quad0\quad&\quad -\ii c k t\sqrt{\alpha }\quad&\quad \ii ckt \sqrt{\alpha}\\
    &\text{multiplicity,}\ m_i&\quad1\quad&\qquad\quad2\quad&\quad2\ .
  \end{aligned}
\]
Symbolic coefficients of 5-th degree minimal polynomial, therefore, have values
\[
\{c_5 \mapsto 0,
c_4 \mapsto c^4 k^4 t^4 \alpha^2,
c_3 \mapsto 0,
c_2 \mapsto 2 c^2 k^2 t^2 \alpha,
c_1 \mapsto 0,
c_0 \mapsto 1\}
\]
which we substitute into formal expression for $S(x,\lambda)$
$$
\begin{aligned}
  S(x,\lambda)&= c_0x^4+ \bigl(c_0 \lambda +c_1\bigr)x^3+ \left(c_0 \lambda ^2+c_1 \lambda +c_2)\right)x^2\\
  &\phantom{=}+\left(c_0 \lambda ^3+c_1 \lambda ^2+c_2 \lambda +c_3\right)x+\left(c_0 \lambda ^4+c_1 \lambda ^3+c_2 \lambda ^2+c_3 \lambda +c_4\right)
 \,.
\end{aligned}
$$
Computation of generalized spectral basis described in {\bf Theorem}~\ref{MainTheorem} then yields
\[
\begin{aligned}
\Bigg\{&
  \Bigl(\lambda_1=0,\quad\Bigl\{Q_{1,0}(x)=\frac{\left(\alpha  c^2 k^2 t^2+x^2\right)^2}{\alpha ^2 c^4 k^4 t^4}
\Bigr\}\Bigr),\\[1pt]
&
\Bigl(\lambda_2=-\ii ckt\sqrt{\alpha},\quad
  \Bigl\{Q_{2,0}(x)=\frac{x \left(2 x+3 \ii \sqrt{\alpha } c k t\right) \left(\sqrt{\alpha } c k t+\ii x\right)^2}{4 \alpha ^2 c^4 k^4 t^4},
\\
  &\phantom{\Bigl(\lambda_2=-\ii ckt\sqrt{\alpha},\!\qquad}
Q_{2,1}(x)=\frac{x \left(x-\ii \sqrt{\alpha } c k t\right)^2 \left(\sqrt{\alpha } c k t-\ii x\right)}{4 \alpha ^{3/2} c^3 k^3 t^3}
\Bigr\}
  \Bigr),\\[4pt]
&
\Bigl(\lambda_3=\ii ckt\sqrt{\alpha},\quad
  \Bigl\{Q_{3,0}(x)=-\frac{x \left(x+\ii \sqrt{\alpha } c k t\right)^2 \left(2 x-3 \ii \sqrt{\alpha } c k t\right)}{4 \alpha ^2 c^4 k^4 t^4},
\\
&\phantom{\Bigl(\lambda_3=\ii ckt\sqrt{\alpha},\!\qquad }
Q_{3,1}(x)=\frac{x \left(x+\ii \sqrt{\alpha } c k t\right)^2 \left(\sqrt{\alpha } c k t+\ii x\right)}{4 \alpha ^{3/2} c^3 k^3 t^3}
\Bigr\}
\Bigr)
\Bigg\}.
\end{aligned}
\]
The largest multiplicity is $2$, therefore using~\eqref{FormulaForFunction} we compute values of the exponential $\exp (x)$ for all three roots and its first derivative $\exp^\prime (x)$ for $\lambda_2$ and  $\lambda_3$
\begin{equation*}\label{ExplicitExponent}
 \begin{aligned}
   &\exp(x)\rvert_{\mu(x)} =\sum_{i=1}^3\ee^{\lambda_i} Q_{i,0}(x)+\frac{\dd\ee^{x}}{\dd x}\Big\rvert_{\lambda_2}Q_{2,1}(x)+\frac{\dd\ee^{x}}{\dd x}\Big\rvert_{\lambda_3}Q_{3,1}(x)\\
    &=\textstyle -\frac{x^4 \left(\frac{1}{2} \sqrt{\alpha } c k t \sin \left(\sqrt{\alpha } c k t\right)+\cos \left(\sqrt{\alpha } c k t\right)-1\right)}{\alpha ^2 c^4 k^4 t^4}
     +\frac{x^3 \left(\sin \left(\sqrt{\alpha } c k t\right)-\sqrt{\alpha } c k t \cos \left(\sqrt{\alpha } c k t\right)\right)}{2 \alpha ^{3/2} c^3 k^3 t^3}\\
    &\textstyle +\frac{x^2 \left(-\frac{1}{2} \sqrt{\alpha } c k t \sin \left(\sqrt{\alpha } c k t\right)-2 \cos \left(\sqrt{\alpha } c k t\right)+2\right)}{\alpha  c^2 k^2 t^2} \scriptstyle +x \left(\frac{3 \sin \left(\sqrt{\alpha } c k t\right)}{2 \sqrt{\alpha } c k t}-\frac{1}{2} \cos \left(\sqrt{\alpha } c k t\right)\right)+1
    \,.
  \end{aligned}
\end{equation*}
After replacing variables $x^i$ in the polynomial above by corresponding matrix powers $x^i\mapsto (-\ii t H_M)^i$ we obtain the final explicit matrix
\begin{equation*}
\exp(-\ii H_Mt)=  \left(\begin{smallmatrix}
 \cos b & -\frac{c g k t \sin b}{2 \sqrt{\alpha }} & 0 & \frac{\ii g \left(\sin b+b \cos b\right)}{2 \sqrt{\alpha }} & -\ii \sqrt{\alpha } \sin b & 0 \\
 0 & \cos b & 0 & \ii \sqrt{\alpha } \sin b & 0 & 0 \\
 0 & 0 & 1 & 0 & 0 & 0 \\
 0 & \frac{\ii \sin b}{\sqrt{\alpha }} & 0 & \cos b & 0 & 0 \\
 -\frac{\ii \sin b}{\sqrt{\alpha }} & \frac{\ii g \left(\sin b-b \cos b\right)}{2 \alpha ^{3/2}} & 0 & \frac{c g k t \sin b}{2 \sqrt{\alpha }} & \cos b & 0 \\
 0 & 0 & 0 & 0 & 0 & 1 
  \end{smallmatrix}\right)
\end{equation*}
where $b=\sqrt{\alpha } c k t$ and where the constant term of the polynomial, {\it i.e.} $+1$ in this case, was multiplied by unit $6\times 6$ matrix. In is important to stress that to get the answer we only need to compute natural matrix powers, {\it i.e.} no need to compute matrix inverse. On the contrary, the method of Jordan decomposition requires the matrix inverse, which is computed by different specialized method. All calculations essentially were done by simple manipulation of polynomials with coefficients in commutative ring $\bbK$.

\section{Concluding remarks}\label{sec7}

With the help of recursive relations and elementary properties of multiple (repeated) roots, we have demonstrated that the recursively generated polynomials form a generalized spectral basis modulo the minimal polynomial. The presented proof is conceptually simple once one identifies the real meaning of recursion, namely, the recursion permits to construct the primary-component basis,
\[
P_i,\quad (x-\lambda_i)P_i,\quad (x-\lambda_i)^2P_i,\quad \dots
\]
inside the quotient algebra \(K[x]/\mu\). The essential message is the following:
1) the repeated roots of a minimal polynomial represent nilpotent directions, 2) the recursion permits to compute the corresponding projector and nilpotent chain without any need to refer to Jordan matrix form, 3) after substitution of variable, \(x\mapsto A\), these polynomial identities become algebraic identities for the original element \(A\).

This provides both the theoretical justification and the effective symbolic method for evaluation of functions of non-diagonalizable algebra
elements. As far as we know this method is already implemented in {\it Mathematica}~\cite{Wolfram2026} kernel function {\bf MatrixFunction[~]} and can be called using option Method$\to$"SpectralBasis". 

Of course, the presented enhancement does not eliminate limitations that are common to the method of spectral basis. For example, the method requires that the function should be differentiable (for example, absolute value function $\lvert A\rvert $ can't be computed), and be non-singular at all eigenvalues. Also for multivalued functions (like square root or logarithm) this method returns a single answer (that is continuously connected to unity), even in the cases~\cite{Acus2024v2} when the matrix/multivector has multiple (even infinite number) square roots.

\end{document}